\documentclass[11pt]{article}
\usepackage{graphicx}
\usepackage{bm}
\usepackage{amsmath}
\usepackage{amssymb}
\usepackage{amsfonts}
\usepackage{amsthm}
\usepackage{mathtools}
\usepackage{hyperref}
\usepackage{braket}
\usepackage[normalem]{ulem}
\usepackage{wrapfig,enumitem}
\usepackage{cite,bbm}
\usepackage{tikz}
\usepackage[left=1in,right=1in,top=1in,bottom=1in]{geometry}

\hypersetup{colorlinks=true,urlcolor=[rgb]{0,0,0.5},citecolor=[rgb]{0.5,0,0},linkcolor=[rgb]{0,0,0.4}}

\newtheorem{theorem}{Theorem}[section]

\newtheorem{lemma}[theorem]{Lemma}
\newtheorem{corollary}[theorem]{Corollary}
 
\newtheorem{definition}[theorem]{Definition}  
\newtheorem{proposition}[theorem]{Proposition} 

\theoremstyle{definition}
\newtheorem{remark}[theorem]{Remark}

\def\autorefapp#1{\hyperref[#1]{Appendix~\ref{#1}}}

\renewcommand{\title}[1]{\vbox{\center\bf{\Large #1}}\vspace{5mm}}
\renewcommand{\author}[1]{\vbox{\center{#1}}\vspace{5mm}}
\newcommand{\address}[1]{\vbox{\center\em#1}}

\def\tr{{\rm tr}}

\def\C{\mathbb{C}}

\def\Z{\mathbb{Z}}

\def\iden{{\rm Id}}

\def\where{\quad {\rm where} \quad}

\def\and{\quad {\rm and} \quad}

\def\CA{{\cal A}}

\def\ketbra#1{|{#1}\rangle\!\langle{#1}|}

\def\Pr{\mathbb{P}}
\def\Ex{\mathbb{E}}

\begin{document}


\title{Two classes of quantum spin systems that are gapped on any bounded-degree graph}
\author{Nicholas Hunter-Jones${}^{a,b}$ and Marius Lemm${}^c$}

\address{${}^a$Department of Physics, University of Texas at Austin, Austin, TX 78712\\
${}^b$Department of Computer Science, University of Texas at Austin, Austin, TX 78712\\
${}^c$Department of Mathematics, University of T\"ubingen, 72076 T\"ubingen, Germany\\
}

\begin{abstract}
We study translation-invariant  quantum spin Hamiltonians on general graphs with non-commuting interactions either given by (i)  a random  rank-$1$ projection or (ii) Haar projectors. 
For (i), we prove that the Hamiltonian is gapped on any bounded-degree graph with high probability at large local dimension. 
For (ii), we obtain a gap for sufficiently large local dimension.
Our results provide examples where the folklore belief that \textit{typical translation-invariant Hamiltonians are gapped} can be proved, which extends a result  by Bravyi and Gosset from 1D qubit chains with rank-$1$ interactions to general bounded-degree graphs.
We derive the gaps by analytically verifying generalized Knabe-type finite-size criteria that apply to any bounded-degree graph. 
\end{abstract}


\section{Introduction}
The existence of a spectral gap above the ground state sector is a central concept in the study of interacting models of quantum matter. (Here, a spectral gap means a difference between the two lowest eigenvalues that is bounded from below independently of the system size.) First and foremost, existence of a spectral gap underpins the very definition of a topological quantum phase known as quasi-adiabatic evolution \cite{hastings2005quasiadiabatic,bachmann2012automorphic}. Moreover, the existence of a uniform spectral gap severely constrains ground state wave functions, implying that they satisfy (i) exponential decay of correlations (proved in all dimensions \cite{hastings2006spectral,nachtergaele2006lieb}) and (ii) an area law for the entanglement entropy (proved in 1D \cite{hastings2007area,arad2017rigorous} and for 2D frustration-free systems \cite{anshu2021area} and believed to be true in any dimension). The existence of a spectral gap is also an important ingredient for the many-body adiabatic theorem \cite{bachmann2018adiabatic,teufel2020non} and for several approaches to quantum computing such as adiabatic quantum computation \cite{aharonov2008adiabatic,albash2018adiabatic} and measurement-based quantum computation \cite{briegel2009measurement}.  

This central role of the spectral gap has raised the important question of whether spectral gaps are in fact common or uncommon occurrences in models of quantum matter.
\\ 

\noindent
\textit{Question 1: Does a ``typical'' quantum many-body Hamiltonian have a spectral gap?}\\

Of course, the answer depends on the precise realization of ``typical'', i.e., which class of Hamiltonians one considers. Physical intuition suggests that typical quantum many-body Hamiltonians should be well inside a quantum phase and should thus enjoy a spectral gap, but this heuristic does not suggest how this gap materializes mathematically. In fact, the average spectral gap between any two consecutive eigenvalues is exponentially small in the system size, so in this sense a size-independent spectral gap is unusual and must depend on considering the spectral edge near the ground state.

A rigorous formulation of Question 1 was introduced in \cite{lemm2019gaplessness}; see also \cite{bravyi2015gapped,lancien2022correlation,jauslin2022random}: Fix a finite graph $\Lambda$ and a local qudit dimension $d$. On the usual many-body Hilbert space $\bigotimes_{x\in\Lambda}\mathbb C^d$ sample a single random interaction, a rank-$r$ projection $P$, and then construct the Hamiltonian
\begin{equation}\label{eq:introH}
    H =\frac{1}{2} \sum_{\substack{x,y\in\Lambda:\\ x\sim y}}P_{x,y},\qquad \textnormal{with }P_{x,y}=P\otimes \iden_{\Lambda\setminus\{x,y\}}
\end{equation}
where $x\sim y$ denotes nearest-neighbors of the graph $\Lambda$. Since $P$ acts in the same way for every bond $x\sim y$, the Hamiltonians \eqref{eq:introH} are called  ``translation-invariant in the bulk''.\footnote{Note that being ``translation-invariant in the bulk'' is weaker property than commuting with translations. The latter would also require imposing a form of periodic boundary conditions.} Quantum spin Hamiltonians that are   ``translation-invariant in the bulk'' are ubiquitous in condensed-matter physics and quantum information theory.
It turns out that one can address the above question about typicality of gaps for these random translation-invariant Hamiltonians under two assumptions on the parameters: (i) the local qudit dimension $d$ is large and (ii) the interaction rank $r$ is small.  We remark in passing that the importance of large qudit dimension has recently been observed in a variety of quantum information settings, arguably most poignantly for error-correcting codes \cite{brock2025quantum}.
Our assumptions imply that $H_\Lambda$ is frustration-free, which is a necessary (but far from sufficient) assumption for all existing derivations of spectral gaps.\footnote{The degree to which frustration-freeness is crucial is perhaps best evidenced by the fact that the Haldane conjecture for the 1D integer-spin Heisenberg antiferromagnet, a simple and natural but frustrated Hamiltonian, remains wide open since 1983.}

A special case of interest has been to study interaction ranks $r=1$ which occur physically, e.g., in certain Rydberg arrays \cite{lin2020quantum,mizuta2020exact,hudomal2022driving}.
Indeed, for rank-$1$ interactions on qubit chains (i.e., $\Lambda=\{1,\ldots,L\}$ is one-dimensional and the local Hilbert space dimension $d=2$), Bravyi and Gosset \cite{bravyi2015gapped} proved that gaps of Hamiltonians of the form Eq.~\eqref{eq:introH} with rank-$1$ interactions are generic, i.e., they occur with probability $1$ for any natural continuous probability measure over rank-$1$ $P$'s.  In a somewhat orthogonal direction, the work \cite{movassagh2017generic} showed that fully disordered Hamiltonians in which each $P_{x,y}$ is sampled \textit{independently} at random, are gap\textit{less} with probability $1$. These disordered Hamiltonians are very far from translation-invariant and the proof in \cite{movassagh2017generic} also uses a completely different mechanism, so-called Griffiths regions. This further evidences that the answer to the question whether gapped or gapless behavior is generic is subtle. 
Next, the work \cite{lemm2019gaplessness} showed that translation-invariance indeed changes the picture in general also beyond the 1D qubit chains considered in \cite{bravyi2015gapped}. Namely, \cite{lemm2019gaplessness} established that there is a positive probability of having a spectral gap for any qudit chain and any interaction rank $r<d$. The work \cite{lancien2022correlation} obtained spectral gaps with high probability for parent Hamiltonians of random translation-invariant tensor product states on $\mathbb Z$ and $\mathbb Z^2$ in the regime of large local dimension and large bond dimension. The work \cite{jauslin2022random} extended the result of \cite{lemm2019gaplessness} for 1D chains to interaction ranks $r<d^2/4$ and, more importantly, to hypercubic lattices in any dimension, requiring only that the rank is sufficiently small compared to the local dimension.\\ 
  
An open problem has been to prove that a randomly chosen translation-invariant Hamiltonian constructed as in Eq.~\eqref{eq:introH} has a high probability of being gapped. A related open problem has been to obtain a proof that  is robust with respect to the local graph geometry.

In our first main result (\textbf{\autoref{thm:main1}}), we resolve these problems for interaction rank $r=1$. We prove that for a random rank $r=1$ interaction, there is a high probability that Hamiltonians of the form Eq.~\eqref{eq:introH} are gapped on any bounded-degree graph for sufficiently large $d$. The result shows that the discovery of Bravyi-Gosset \cite{bravyi2015gapped} for $r=1$ qubit chains is much more general. One novelty   is universality with respect to the geometrical structure, i.e., the probability only depends on the maximal degree and is otherwise independent of the graph. 
Such robustness of certain spectral gap properties with respect to the graph structure has recently been of growing interest in related contexts \cite{lucia2023nonvanishing,lemm2025critical}. In our case, the robustness arises from a  generalization of Knabe's finite-size criterion \cite{knabe1988energy} to bounded-degree graphs (Proposition \ref{prop:fs}).

In a second, complementary result (\textbf{\autoref{thm:main2}}), we use a similar finite-size criterion (Proposition \ref{prop:fs2}) to derive spectral gaps for a class of deterministic Hamiltonians \eqref{eq:introH} where $P$ is constructed from Haar projectors. These Hamiltonians are related to unitary designs \cite{brown2010convergence,brandao2016local,haferkamp2021improved,mittal2023local,Harrow23}. We prove that that any member of this class of Hamiltonians is gapped on a bounded-degree graph for sufficiently large local dimension. Interestingly, their interactions can have rather large rank. 
These models therefore provide yet another class of explicit gapped Hamiltonians on \textit{arbitrary bounded-degree} graphs. The second result shows the versatility of the finite-size criteria we consider here, because it provides a second completely different class of models with large interaction ranks where they can be verified.

We recall undecidability results for the spectral gap of translation-invariant Hamiltonian \cite{cubitt2015undecidability,cubitt2022undecidability}, which also require large local qudit dimension. While our result proves a gap, and thus decidability, it does not stand in any contradiction to these results. Naturally, our results limit the possible scope of undecidability. E.g., Theorem \ref{thm:main1} implies as a corollary that undecidability of the gap problem cannot occur for a typical $r=1$ interaction. 

%



\section{Main results}

\subsection{Result 1: Gapped random Hamiltonians}
We present a class of random translation-invariant Hamiltonians. 
Let $\Lambda$ be a graph of maximal degree $k$. Fix a projection operator $P$ on $\mathbb C^d\otimes \mathbb C^d$ of rank $r=1$. We consider Hamiltonians of the form \eqref{eq:introH}, i.e.,
\begin{equation}\label{eq:HLambda}
H_\Lambda=\frac{1}{2}\sum_{\substack{x,y\in\Lambda:\\ x\sim y}} P_{x,y}\,, \where P_{x,y}=P\otimes \iden_{\Lambda\setminus\{x,y\}}\,.
\end{equation}

For definiteness, we take the single random interaction $P$ to be sampled as follows. We let $G=G^T\in \mathbb R^{d\times d}$ be a \textit{GOE random matrix}, i.e., we take independent Gaussians
\[
\begin{aligned}
  G_{i,j}\sim& \mathcal N_{\mathbb R}(0,\tfrac{1}{d}),\qquad 1\leq i<j\leq d,\\ 
  G_{i,i}\sim&\sqrt{2} \mathcal N_{\mathbb R}(0,\tfrac{1}{d}),\qquad 1\leq i\leq d
\end{aligned}
\]  
and set $G_{i,j}=G_{j,i}$ for $i>j$.
 We normalize the GOE matrix by
\begin{equation}\label{eq:Cdefn}
C=\frac{1}{\sqrt{\mathrm{Tr} (G^\dagger G)}} G
\end{equation}
and define the vector $v\in\mathbb C^d\otimes \mathbb C^d$ as
\[
v=\sum_{i,j=1}^dC_{ij}e_i\otimes e_j,\qquad \|v\|^2=\mathrm{Tr}( C^\dagger C)=1.
\]
Finally, we set
\[
P=\ketbra{v}\,.
\]
In summary, for a single $d\times d$ GOE random matrix $G$, this defines a random translation-invariant Hamiltonian $H_\Lambda$ on any graph $\Lambda$ via \eqref{eq:HLambda}.      We remark that  $P$ commutes with the swap operator, i.e., the edges are not oriented. This follows from
 $G=G^T$ and it is the reason why we choose the GOE class instead of, say, the GUE class. This choice is made for convenience and the method can be adapted to non-symmetric coefficient matrices (and thus oriented edges) as described in Remark \ref{rmk:generalizations}.

We recall the following two fundamental definitions.

\begin{definition}[Frustration-freeness]
We say the Hamiltonian $H_\Lambda$ in \eqref{eq:introH} is frustration-free, if $\ker(H_\Lambda)\neq\{0\}$.
\end{definition}
By the QSAT criterion in Lemma 2.4 from Ref.~\cite{jauslin2022random} which follows from \cite{sattath2016local} (see also \cite{movassagh2010unfrustrated}) the Hamiltonian $H_\Lambda$ is frustration-free, if 
\begin{equation}\label{eq:FFconstraint}
    d^2\geq e(2k-1).
\end{equation}

\begin{definition}[Spectral gap]
    For $H_\Lambda$ frustration-free, we write $\Delta_\Lambda$ for its smallest positive eigenvalue and we call it the spectral gap.
\end{definition}

Our first main result says that there is a positive probability that the so-constructed Hamiltonians $H_\Lambda$ are gapped on any degree-$k$ graph, for any fixed $k$. This probability converges exponentially to $1$ as the local qudit dimension $d$ increases. The result gives precise quantitative parameter ranges. Since it uses a concentration bound for random matrices, a key role is played by the smallest moment of a $d\times d$ GOE matrix.
\begin{equation}\label{eq:mddefn}
m_d=\min_{p\geq 1} ( \mathbb{E}[\tr(G^{2p})])^{1/(2p)}
\end{equation}
There are many ways to obtain good analytical estimates on $m_d$, as we describe in Prop.\ \ref{prop:md} below.

\begin{theorem}[Main result 1] \label{thm:main1}
Let $d,k\geq 1$ be integers and let $\delta>0$ be such that \eqref{eq:FFconstraint} holds and
\begin{equation}\label{eq:main1condition}
\frac{d(1-\delta)}{(m_d^2+\delta/2)^2}>2k-2.
\end{equation}

 Then, there exists a constant $c>0$ and an event $\Omega$ with probability at least
\begin{equation}
    \mathbb P(\Omega)\geq 1-2e^{-d^2\delta^2 /4}
\end{equation}
such that for every $\omega\in\Omega$ the following holds: For every finite graph $\Lambda$ of maximal degree $k$,
\[
\Delta_\Lambda>c>0.
\]
\end{theorem}

\begin{remark}\label{rmk:main1}
    \begin{enumerate}[label=(\roman*)]
    \item The constant $c$ is also explicit and can be read off from the proof. 
      
       \item For $k=2$, the condition on $d$ to obtain a gap with probability at least  99\% is $d\geq 15$. For $k=3$, it is $d\geq 24$. 
       \item Multi-level optical cavity experiments have reached local dimensions up to $d=12$ \cite{sundar2024driven,PhysRevResearch.6.L032072}. 
        
    \end{enumerate}
\end{remark}

In order to estimate $m_d$, we provide some facts about moments of GOE matrices.

\begin{proposition}[Bounds on $m_d$]\label{prop:md}
    We have
\begin{equation}\label{eq:Ledoux}
      \begin{aligned}
    \mathbb{E}[\tr(G^{2})] &= d^2 + d\\
      \mathbb{E}[\tr(G^{4})] &= 2 d^3 + 5 d^2 + 5 d\\
        \mathbb{E}[\tr(G^{6})] &= 5d^4 +22 d^3 +52 d^2 +41d\\
      \mathbb{E}[\tr(G^{8})] &= 14d^5 + 93d^4 + 374d^3 + 690d^2 + 509d
    \end{aligned}
\end{equation}
and generally for any $d>1$
\begin{equation}\label{eq:mdbound}
      \begin{aligned}
      m_d\leq 2e\sqrt{\lceil\log d\rceil}.
   \end{aligned}
\end{equation}
\end{proposition}

Theorem \ref{thm:main1} and \eqref{eq:mdbound} are proved in Section \ref{sec:proof}. 
The identities \eqref{eq:Ledoux} are taken from a paper of Ledoux \cite{ledoux2009recursion} who derives them from an explicit recursion formula for the moments of a GOE matrix. His recursion formula can also be used to calculate many higher moments. Notice that the leading coefficients in \eqref{eq:Ledoux} are the Catalan numbers as they have to be.

The bounds on $m_d$ obtained from \eqref{eq:Ledoux} work better in smaller dimensions $d$, whereas \eqref{eq:mdbound} is clearly stronger for large dimension. In particular, for large local dimension $d$ and degree $k$, we can use \eqref{eq:mdbound} on \eqref{eq:main1condition} to obtain the asymptotic condition
\[
\frac{d}{\log d}\gtrsim \frac{8e^2}{1-\delta} k
\]
This means that, up to log-corrections, the critical local dimensions to ensure a gap with high probability grows linearly in the degree of the graph.



\begin{remark}[Generalizations]\label{rmk:generalizations}
    Above, we choose the coefficient matrix $G$ to be a GOE matrix, which is the natural choice for a generic symmetric, real-valued matrix. The symmetry of $G$ is equivalent to $P$ commuting with the swap operator, i.e., to edges being undirected and the physical meaning of $G$ being real-valued is that it is time-reversal invariant.  Both of these assumptions can be removed if desired and many other generalizations of our proof are possible. For example, one could replace $G$ by a GUE or CUE matrix, which would make it neither symmetric nor real-valued. To adapt the proof, one can further generalize the Knabe-type argument to directed graphs of bounded degree. Moreover, using random matrix universality  results, one could change the coefficient matrix $G$ to be non-Gaussian. This would mainly affect the quantitative estimates for small $d$-values.
\end{remark}

\subsection{Result 2: Gapped Hamiltonians from Haar projectors}
Whereas our first result \autoref{thm:main1} holds with high probability for random rank-$1$ Hamiltonians, we can use similar methods to prove that a class of deterministic Hamiltonians, again defined on any bounded-degree graph, is also gapped for sufficiently large local dimension. 
The explicit construction is based on Haar projectors.

The models depend on two independent integer parameters $q,t\geq 1$. We focus on the regime $q>t$, in which case the relation to the local dimension $d$ and interaction rank $r$ from the previous sections will be as follows:
\[
d=q^{4t},\qquad r=q^{4t}-t!
\]
Since we consider large $q$-values, we see that the local interactions will not only have large local dimension $d$ but in fact large rank $r$, in contrast to the random rank-$1$ models we considered before.

The local Hilbert spaces are $\C^{q^{2t}}$, i.e., $2t$ copies of a single $q$-dimensional Hilbert space. As before, we let $\Lambda$ denote a finite graph and consider a Hamiltonian of the form \eqref{eq:introH}, i.e.,
\begin{equation}\label{eq:HLambda2}
    H_\Lambda = \sum_{\substack{x,y\in \Lambda\\ x\sim y}} P_{x,y}\,, \where P_{x,y} = P\otimes \mathrm{Id}_{\Lambda\setminus \{x,y\}},
\end{equation}
where  $P$ acts on $\C^{q^{2t}}\otimes \C^{q^{2t}} $ and is defined as
\begin{equation}\label{eq:Qdefn}
    P=Q^\perp,\qquad Q = \int_{\mathrm{U}(q^2)}  U^{\otimes t}\otimes \overline{U}^{\otimes t} \mathrm{d}\mu_H(U).
\end{equation}
with $ \mathrm{d}\mu_H$ the Haar measure over the unitary group $\mathrm{U}(q^2)$.
The fact that $Q^2=Q$ (and thus also $P^2=P$) follows from the left/right invariance of the Haar measure.

We remark that the Hamiltonian \eqref{eq:HLambda2} defined via Haar projectors $Q$ in \eqref{eq:Qdefn} is motivated by the study of random quantum circuits and unitary designs, where  one is interested in the convergence of random circuits to approximate unitary $t$-designs, distributions on the unitary group which equal or are close to the Haar measure up to the $t$-th moment. One standard approach to this is to study the gap of the $t$-th moment operator of the distribution, which controls the convergences \cite{brandao2016local,haferkamp2021improved}. 

\begin{lemma}[Lemma 17 in \cite{brandao2016local}]\label{lm:2ff}
    $H_\Lambda$ defined by \eqref{eq:HLambda2} and \eqref{eq:Qdefn} is frustration-free.
\end{lemma}



Our second main result gives parameter regimes when $H_\Lambda$ defined by \eqref{eq:HLambda2} and \eqref{eq:Qdefn} is gapped.

\begin{theorem}[Second main result]
\label{thm:main2}
Let $\Lambda$ be a finite graph of maximal degree $k$. Then $H_\Lambda$ defined by \eqref{eq:HLambda2} and \eqref{eq:Qdefn} satisfies
\[
\Delta_\Lambda \geq  1-\frac{4(k-1)t^2}{q}.
\]
\end{theorem}

Hence, for $q>4(k-1)t^2$, the gap is bounded from below uniformly.
These results are proved in Section \ref{sec:2proof}.

 \subsection{Application to the 2D area law}\label{ssect:ALapplication}
 As an application of our results, we obtain random Hamiltonians with ground states satisfying the area law with high probability. We recall that verifying the 2D area law is notoriously difficult in the context of PEPS techniques. In our case, the key insight is that the finite-size criterion automatically lower bounds even the ``local gap'' which is the key input to the area law criterion for 2D grids by Anshu, Arad, and Gosset \cite{anshu2021area}. (A difference to the setting considered therein is that the ground state space of our random spin chain Hamiltonians is highly degenerate and this point was already addressed in \cite{jauslin2022random}.)  Consequently, our result gives new examples of 2D Hamiltonians that provably have ground states satisfying an area law, which are generally difficult to construct by other methods.

We recall the setup for the area law. Let $L\geq 2$ and take the graph $\Lambda=((-L,L]\cap \Z)^2$ to be a box in $\Z^2$. We split the Hilbert space into left and right halves in the natural way
$$
\mathcal H=\mathcal H_{\mathrm{left}}\otimes \mathcal H_{\mathrm{right}}.
$$
Given a pure ground state $\psi$ of $H_\Lambda$, we consider the marginal on the left, given by taking the partial trace over the right side,
$$
\rho_{\psi}^{\mathrm{left}}=\mathrm{Tr}_{\mathcal H_{\mathrm{right}}}\vert\psi\rangle\langle\psi\vert
$$
By definition, the entanglement entropy of the left and right halves equals
\begin{equation}
S(\rho_{\psi}^{\mathrm{left}})=-\mathrm{Tr}\rho_{\psi}^{\mathrm{left}}\log\rho_{\psi}^{\mathrm{left}}.
 \end{equation}



\begin{corollary}[2D area law]\label{cor:AL}
In the settings of Theorem \ref{thm:main1} and Theorem \ref{thm:main2}, specialized to $\Lambda=((-L,L]\cap \Z)^2$ a box in $\Z^2$, we have that for every $\omega\in \Omega$ $H_\Lambda$ has a ground state satisfying the area law
\begin{equation}\label{eq:AL}
S(\rho_{\psi}^{\mathrm{left}})\leq cL^{1+(\log L)^{-1/5}}.
 \end{equation}
\end{corollary}

We remark that the mild logarithmic correction in \eqref{eq:AL} is unimportant and can be ignored for all practical purposes.

The proof of Corollary \ref{cor:AL} is completely analogous to the proof of Corollary 2.10 in \cite{jauslin2022random} and we therefore omit it. In a nutshell, the idea is that the proof of the finite-size criteria Proposition \ref{prop:fs} and Proposition \ref{prop:fs2}  in fact yields a lower bound on the so-called local gap (cf., Def.\ 2.5 and Prop.\ 2.11 in \cite{jauslin2022random}). The lower bound on the local gap can then be fed into the area law criterion from \cite{anshu2021area}; more specifically its analog for degenerate ground states, Theorem 2.8 in \cite{jauslin2022random}, which is due to \cite{anshu}. We leave the details to the reader.

\subsection{Proof strategy}
To prove \autoref{thm:main1}, we proceed as follows. 
First, we ensure frustration-freeness by employing a QSAT (quantum satisfiability) criterion of \cite{sattath2016local}. This requires as input a classical cluster expansion bound that follows from the Kotecky-Preiss criterion, cf.\ \cite{jauslin2022random}. 
Thanks to frustration-freeness, we can then employ the method of finite-size criteria. A finite-size criterion  says that if the gap of suitably small subsystems exceed some explicit threshold, the Hamiltonians are gapped uniformly in the system size.   Finite-size criteria have recently been used to derive spectral gaps in random Hamiltonians \cite{lemm2019gaplessness,lancien2022correlation,jauslin2022random}  and a number of other frustration-free models  \cite{lemm2019gapped,pomata2019aklt,abdul2020class,lemm2020existence,pomata2020demonstrating,nachtergaele2021spectral,haferkamp2021improved,andrei2022spin,warzel2023bulk}. Here, we generalize Knabe's original criterion \cite{knabe1988energy} to general graphs of bounded degree; see Propositions \ref{prop:fs} and \ref{prop:fs2}. See \cite{gosset2016local,lemm2019spectral,lemm2019gapped,lemm2020finite,lemm2020existence,nachtergaele2021spectral,lemm2022quantitatively} for other improvements of Knabe's criterion. We verify the finite-size condition analytically. This relies on certain algebraic identities that relate the angle between neighboring projections to the entanglement of the rank-$1$ vector $v$, as well as explicit concentration bounds for random matrices. For the latter, we take care to avoid asymptotic results with non-explicit constants. Instead, we use hands-on estimates in order to conclude existence of a gap for Hamiltonians with explicit, physically meaningful parameter values. Specifically, we care about obtaining reasonable estimates on the probability of having a gap also for relatively small values of the local qudit dimension $d$.

To prove \autoref{thm:main2}, we use a slightly different approach. Again, we use a generalized Knabe finite-size criterion to relate the $n$-site Hamiltonian gap to the gap of the Hamiltonian on star graphs which is due to \cite{mittal2023local}. Their ground space can be understood in terms of permutation operators. We may then verify the finite-size criterion using the approximate orthogonality of permutation operators in large dimension \cite{brandao2016local,haferkamp2021improved}.

The overarching idea of our paper is that the robustness with respect to the graph structure is a consequence of using a highly local finite-size criterion. The locality is what makes the finite-size criterion rough, but robust.

\section{Finite-size criteria on bounded-degree graphs}
In the present section, we consider general Hamiltonians of the form \eqref{eq:HLambda}, i.e.,
\[
H_\Lambda=\frac{1}{2}\sum_{\substack{x,y\in\Lambda:\\ x\sim y}} P_{x,y}
\]
We present two finite-size criteria in the spirit of Knabe \cite{knabe1988energy} assuming  that $H_\Lambda$ is frustration-free and that $\Lambda$ has bounded maximal degree $k$.  Both criteria provide lower bounds on the spectral gap $\Delta_\Lambda$ which only depends on the maximal degree $k$, not on any other details of the graph. The latter criterion already appeared in \cite{mittal2023local} and we include the short proof here. Both proofs are direct generalizations of Knabe's combinatorial argument \cite{knabe1988energy} to non-constant vertex degree. 

Later in the paper, we will apply the first finite-size criterion  to prove our first main result, Theorem \ref{thm:main1}, and we will apply the second finite-size criterion to prove our second main result, Theorem \ref{thm:main2}. In both applications, frustration-freeness is ensured by a different argument. 

For notational convenience, in this section, all sums run over $\Lambda$ unless otherwise stated. Given a fixed vertex $x$, we write $\sum_{y\sim x}$ for the sum over nearest-neighbors $y$ of $x$.

\subsection{Criterion 1: The two-leg criterion}
Let $\Delta_\Lambda$ denote the gap of $H_\Lambda$. Let $\Delta_3$ denote the gap of the ``two-leg'' Hamiltonian
 \[
 H_3=P_{1,2}+P_{2,3}
 \]
Notice that the frustration-freeness of any $H_\Lambda$ of degree $\geq 2$ implies that $H_3$ is frustration-free.

\begin{proposition}
[Two-leg criterion]
\label{prop:fs} Let $k\geq 2$. Let $\Lambda$ be a graph of maximal degree $k$ such that $H_\Lambda$ is frustration-free. Then
\begin{equation}
    \Delta_\Lambda \geq 2(k-1) \left(\Delta_3-\frac{2k-3}{2k-2}\right).
\end{equation}
\end{proposition}

\begin{proof}[Proof of Proposition \ref{prop:fs}]
We square the Hamiltonian and drop non-negative terms, i.e.\ the anti-commutators of non-overlapping terms, to find the operator inequality
\[
H_\Lambda^2\geq H_\Lambda+\sum_x \sum_{\substack{y,z\sim x\\ y\neq z}} \{P_{x,y},P_{y,z}\}
\]
with $\{A,B\}=AB+BA$ the anti-commutator. We introduce the auxiliary operator
\[
\mathcal A =\sum_x \sum_{\substack{y,z\sim x\\ y\neq z}} (P_{x,y}+P_{y,z})^2
\]
On the one hand, we can calculate
\[
\mathcal A =2\sum_x \left((\mathrm{deg}(x)-1) \sum_{y\sim x} P_{x,y}+\sum_{\substack{y,z\sim x\\ y\neq z}}  \{P_{x,y},P_{y,z}\}\right)
\]
On the other hand, by frustration-freeness of $H_3$, we have the operator inequality
\[
\mathcal A \geq 2\Delta_3 \sum_x (\mathrm{deg}(x)-1) \sum_{y:y\sim x} P_{x,y}
\]
and so
\begin{align}
H_\Lambda^2\geq&
\sum_{x\sim y} P_{x,y}
\left(
1-2(1-\Delta_3)  (\mathrm{deg}(x)-1) 
\right)\\
\geq&(1-2(1-\Delta_3)(k-1))H_\Lambda\\
=&2(k-1)\left(\Delta_3-\frac{2k-3}{2k-2}\right)H_\Lambda
\end{align}
The claimed spectral gap inequality now follows because $H_\Lambda$ is frustration-free.
\end{proof}

\subsection{Criterion 2: The star graph criterion}
The following star-graph criterion is different from Proposition \ref{prop:fs} in two ways: The good news is that the threshold, $\tfrac{1}{2}$, is lower, but the bad news is that the gap on $3$ sites $\Delta_3$ is replaced by  the minimal gap on star graphs whose size may reach up to $k$, the degree of the original graph. 

For $m\geq 2$, we let $\Delta_m$ denote the spectral gap on the $m$-site star graph, i.e., $\Delta_{m}$ is the spectral gap of
\[
H_m=P_{1,2}+P_{1,3}+\ldots+P_{1,m}.
\]

\begin{proposition}[Star graph criterion \cite{mittal2023local}]\label{prop:fs2}
 Let $k\geq 2$. Let $\Lambda$ be a graph of maximal degree $k$ such that $H_\Lambda$ is frustration-free. Then
\begin{equation}
    \Delta_\Lambda \geq 2\bigg( \min_{2\leq m\leq k} \Delta_{m} - \frac{1}{2}\bigg)\,,
\end{equation}
\end{proposition}

This is essentially Theorem 1 in \cite{mittal2023local} and we include the short proof for completeness. Note that $\Delta_2=1$.

\begin{proof}[Proof of Proposition \ref{prop:fs2}]
We square the Hamiltonian and drop non-negative terms, i.e.\ the anti-commutators of non-overlapping terms, to find the operator inequality
\begin{equation}\label{eq:refabove1}
H_\Lambda^2\geq H_\Lambda+\sum_x \sum_{\substack{y,z\sim x\\ y\neq z}} \{P_{x,y},P_{y,z}\}.
\end{equation}
 However, we now define a different auxiliary subsystem operator by squaring the sum over all the incident edges at a vertex, i.e.,
\begin{equation}\label{eq:refabove2}
    \CA = \sum_{x} \left( \sum_{y\sim x} P_{x,y}\right)^2 = 2H_\Lambda +\sum_x \sum_{\substack{y,z\sim x:\\ y\neq z}}\{P_{x,y},P_{y,z}\}.
\end{equation}
 We now expand the operator $\mathcal{A}$ as a sum over all vertices of a fixed degree
\begin{equation}
    \mathcal{A} = \sum_{k'=1}^k \sum_{\substack{x: \deg(x)=k'}} \left( \sum_{y\sim x} P_{x,y}\right)^2\geq \sum_{k'=1}^k  \Delta_{k'}\sum_{\substack{x: \deg(x)=k'}} \sum_{y\sim x} P_{x,y}.
\end{equation}
We then find that
\begin{equation}
    \mathcal{A} \geq 2H_\Lambda \min_{1\leq k'\leq k}\Delta_{k'}.
\end{equation}
The claim now follows by combining this with \eqref{eq:refabove1} and \eqref{eq:refabove2}.
\end{proof}

\section{Proof of \autoref{thm:main1}}\label{sec:proof}
In this section, we work in the setting of \autoref{thm:main1} of random rank-$1$ interactions.
 
\subsection{Local gap bound}
Considering Proposition \ref{prop:fs}, our goal is to derive a lower bound on $\Delta_3$. This is the content of the following proposition, whose proof will occupy most of this section.

\begin{proposition}[Lower bound on $\Delta_3$]
\label{prop:Delta3}
 We have
 \begin{equation}
     \Delta_3 \geq 1-\frac{(m_d+\epsilon)^2}{d(1-\delta)}
\end{equation}
holds with probability at least 
\[
1-e^{-d^2\delta^2 /4}- e^{-d^2\epsilon^2}.
\]
\end{proposition}

To lower bound $\Delta_3$, we begin by squaring $H_3$. Using the standard Fannes-Nachtergaele-Werner bound, Lemma 6.3 in \cite{fannes1992finitely}, we obtain
    \begin{equation}\label{eq:Delta3start}
            H_3^2=H_3+\{P_{1,2},P_{2,3}\}\geq (1-\|P_{1,2}P_{2,3}-P_{1,2}\wedge P_{2,3}\|)H_3,
   \end{equation}
    where $P_{1,2}\wedge P_{2,3}$ denotes the orthogonal projection onto $\mathrm{ran}\, P_{1,2}\cap \mathrm{ran}\, P_{2,3}$.

    Hence, Proposition \ref{prop:Delta3} will follow from the norm bound
\begin{equation}\label{eq:normbound1}
    \|P_{1,2}P_{2,3}-P_{1,2}\wedge P_{2,3}\|\leq \frac{(m_d+\epsilon)^2}{d(1-\delta)}.
\end{equation}

The proof of \eqref{eq:normbound1} uses two lemmas:
Lemma \ref{lemma:intersection} proves that $P_{1,2}\wedge P_{2,3}=0$ almost surely and  Lemma \ref{lemma:exactcalc} calculates $\|P_{1,2}P_{2,3}\|$ resulting in certain random matrix norms.
We recall that $P=\ketbra{v}$ where $v\in\mathbb C^d\otimes \mathbb C^d$ is given by
\[
v=\sum_{i,j=1}^dC_{ij}e_i\otimes e_j,\qquad \|v\|^2=\mathrm{Tr}( C^\dagger C)=1.
\]

\begin{lemma}\label{lemma:intersection}
$P_{1,2}\wedge P_{2,3}\neq \{0\}$ holds if and only if $ v=u\otimes u$ for some $u\in\mathbb C^d$.
\end{lemma}

\begin{proof}[Proof of Lemma \ref{lemma:intersection}]
We have
\[
\begin{aligned}
\mathrm{ran}\,    P_{1,2}=&
\mathrm{span}\left\{\ket{v}_{12}\otimes \ket{k}_3\,:\, 1\leq k\leq d\right\}
=\mathrm{span}\left\{\sum_{i,j=1}^d C_{ij}\ket{ijk}\,:\, 1\leq k\leq d\right\},\\
\mathrm{ran}\,    P_{2,3}=&
\mathrm{span}\left\{\ket{i}_1\otimes \ket{v}_{23}\,:\, 1\leq i\leq d\right\}
=\mathrm{span}\left\{\sum_{j,k=1}^d C_{jk}\ket{ijk}\,:\, 1\leq i\leq d\right\}.\\
\end{aligned}
\]
By definition, $P_{1,2}\wedge P_{2,3}\neq \{0\}$ is equivalent to $\mathrm{ran}\,P_{1,2}\cap \mathrm{ran}\,P_{2,3}\neq\{0\}$. This holds iff there exist not identically zero coefficient vectors $\lambda=(\lambda_1,\ldots,\lambda_d)$ and $\mu=(\mu_1,\ldots,\mu_d)$ in $\mathbb C^d$ such that
\[
\sum_{i,j,k=1}^d \lambda_k C_{ij}\ket{ijk}
=\sum_{i,j,k=1}^d \mu_i C_{jk}\ket{ijk}.
\]
By orthonormality, this implies
\[
\lambda_k  C_{ij}=\mu_i C_{jk},\qquad \forall i,j,k\in\{1,\ldots,d\}.
\]
Let $k_0\in \{1,\ldots,d\}$ be the minimal choice such that $\lambda_{k_0}\neq 0$. Then
\[
C_{ij}=\mu_i \frac{  C_{jk_0}}{\lambda_{k_0}} ,\qquad \forall i,j\in\{1,\ldots,d\}.
\]
Therefore,
\[
\begin{aligned}
\ket{v}=&\sum_{i,j=1}^d C_{ij} \ket{ij}
=\sum_{i,j=1}^d \mu_i \frac{  C_{jk_0}}{\lambda_{k_0}} \ket{ij}
=\ket{u_1}\otimes \ket{u_2},\\
\textnormal{with }
\ket{u_1}=&\sum_{i=1}^d \mu_i \ket {i},
\qquad \ket{u_2}=\sum_{j=1}^d \frac{  C_{jk_0}}{\lambda_{k_0}}\ket{j}.
\end{aligned}
\]
Since $C$ is symmetric, we also have
\[
\mu_i \frac{  C_{jk_0}}{\lambda_{k_0}} =\mu_j \frac{  C_{ik_0}}{\lambda_{k_0}},\qquad \forall i,j\in\{1,\ldots,d\}.
\]
This implies $\ket{v}=\ket{u_2}\otimes \ket{u_1}$ and hence $\ket{u_1}=\ket{u_2}$. To check that the reverse implication holds as well, we proceed through the argument in reverse order.
\end{proof}

\begin{lemma}\label{lemma:exactcalc}
    It holds that
    \[
\|P_{1,2}P_{2,3}\|=\frac{\|G\|^2}{\mathrm{Tr}(G^2)}
\]
\end{lemma}

\begin{proof}
Recall that 
\[
\begin{aligned}
  P_{1,2}=&
\ket{v}\bra{v}_{12}\otimes \iden_3
=
\sum_{i,j,k,l,m=1}^d C_{ij} C_{kl}\ket{ijm}\bra{klm}\\
  P_{2,3}=&
\iden_1\otimes \ket{v}\bra{v}_{23} 
=\sum_{n,p,q,r,s=1}^d C_{pq} C_{rs}\ket{npq}\bra{nrs}.
\end{aligned}
\]
Therefore
\[
\begin{aligned}
P_{1,2}P_{2,3}
=&\sum_{i,j,k,l,m=1}^d  \sum_{n,p,q,r,s=1}^d
C_{ij} C_{kl} C_{pq} C_{rs}\delta_{k,n}\delta_{l,p}\delta_{m,q} \ket{ijm}\bra{nrs}\\
=&\sum_{i,j,k,m=1}^d \sum_{r,s=1}^d
C_{ij} (C^2)_{km}  C_{rs} \ket{ijm} \bra{krs}.
\end{aligned}
\]
We use unitary invariance of the operator norm to suitably relabel the matrix columns. Namely, we define the  operator
\[
U\ket{krs}=\ket{rsk}
\]
and extend it to a unitary $U$ on $\mathbb C^d\otimes \mathbb C^d\otimes \mathbb C^d$ by linearity. Then

\[
P_{1,2}P_{2,3}U
=\sum_{i,j,k,l,m=1}^d \sum_{r,s=1}^d
C_{ij} (C^2)_{km}  C_{rs} \ket{ijm} \bra{rsk}
=\ket{v}\bra{v}_{12}\otimes C^2
\]
and so
\[
\|P_{1,2}P_{2,3}\|
=\|P_{1,2}P_{2,3} U\|=\|v\|_2 \|C\|^2=\|C\|^2.
\]
The claim now follows from Definition \eqref{eq:Cdefn} of $C$ in terms of $G$.
\end{proof}


\begin{proof}[Proof of Proposition \ref{prop:Delta3}]
    We begin with \eqref{eq:Delta3start}. Hence, it remains to bound $ \|P_{1,2}P_{2,3}-P_{1,2}\wedge P_{2,3}\|$. We restrict to the event that $P_{1,2}\wedge P_{2,3}=0$, which can be done without loss of generality because this event has probability $1$ by Lemma \ref{lemma:intersection}.  

    On this event, by Lemma \ref{lemma:exactcalc} 
    \[
    \|P_{1,2}P_{2,3}-P_{1,2}\wedge P_{2,3}\|=\|P_{1,2}P_{2,3}\|\leq \frac{\|G\|^2}{\mathrm{Tr}(G^2)}.
    \]
We now bound the numerator and the denominator separately by probabilistic tail bounds. 

We start with the tail bound for the denominator, which is simpler. The mean is 
\[
\mathbb E\left[ \mathrm{Tr}(G^2)\right]=\sum_{i,j}\mathbb E[G_{ij}^2]=d.
\]
Since $\frac{1}{d}\mathrm{Tr}(G^2)=\tfrac{1}{d}\sum_{i,j=1}^d G_{ij}^2$ is a sum of i.i.d.\ random variables we can obtain tail bounds from the exponential Markov inequality. Let $\delta>0$. Then, with $t>-\frac{d^2}{2}$ a free parameter,
\[
\begin{aligned}
&\mathbb P\left(\tfrac{1}{d}\mathrm{Tr}(G^2)<1-\delta\right)\\
=&\mathbb P\left(-\tfrac{1}{d}\mathrm{Tr}(G^2)>-(1-\delta)\right)\\
\leq& \exp\left(\log\mathbb E\left[\exp(-\tfrac{t}{d}\mathrm{Tr}(G^2))\right]+(1-\delta)t\right)\\
=& \exp\left(-\tfrac{d^2}{2}\log \left(1+2\tfrac{t}{d^2}\right)+(1-\delta)t\right)
\end{aligned}
\]
In the last step we used that for a sum $S_n=\sum_{i=1}^n Z_i^2$ of independent standard Gaussians $Z_1,\ldots,Z_n$ the moment generating function is $\mathbb E[e^{\tau S_n}]=(1-2\tau)^{-n/2}$ for $\tau<\tfrac{1}{2}$. The optimal choice of $t$ is $t_\delta=\frac{\delta}{1-\delta} \frac{d^2}{2}$ which gives the tail bound
\[
\mathbb P\left(\tfrac{1}{d}\mathrm{Tr}(G^2)<1-\delta\right)\leq \exp\left(\frac{d^2}{2} \left(\log(1-\delta)+\delta\right)\right)
 \leq e^{-d^2\delta^2/4}.
\]

We come to the tail bound for the numerator. For this we use L\'evy's lemma (or earlier Gaussian concentration bounds) and the fact that the operator norm is 1-Lipschitz. For this, we recall that the standard deviation of each entry of $G$ is $\frac{1}{d}$.
Let $\epsilon>0$. Then, 
\begin{equation}\label{eq:concentration}
\begin{aligned}
        &\Pr\Big( \| G\| \geq m_d+\epsilon\Big)
    \leq \Pr\left( \| G\| \geq \sqrt{\Ex\big[\| G^2\|\big]}+\epsilon\right)
    \leq \Pr\Big( \| G\| \geq \Ex\big[\| G\|\big] + \epsilon\Big) \leq e^{-d^2\epsilon^2}.
\end{aligned}
\end{equation}
By the monotonicity of Schatten $p$-norms and Jensen's inequality, we have
for any $p\geq 1$
\[
    \Ex\big[\| G^2\|]\leq \Ex\big[\| G^2\|_p]\leq \left(\Ex\big[\tr(G^{2p})]\right)^{1/p}
\]
Minimizing over $p$ yields
\[
 \Ex\big[\| G^2\|]\leq m_d^2
\]
and so
\begin{equation}
    \Pr\Big( \| G\|^2 \leq (m_d+\epsilon)^2 \Big)\geq 1- e^{-d^2\epsilon^2}\,.
\end{equation}

We can now combine the tail estimates via a union bound. We obtain that
\begin{equation}
     \frac{\|G^2\|}{\tr(G^2)} \leq \frac{(m_d+\epsilon)^2}{d(1-\delta)}
\end{equation}
holds with probability at least 
\[
1-e^{-d^2\delta^2 /4}- e^{-d^2\epsilon^2}.
\]
This proves Proposition \ref{prop:Delta3}.
\end{proof}

\subsection{Conclusion}
\begin{proof}[Proof of Theorem \ref{thm:main1}]
By Proposition \ref{prop:Delta3}, we have
\[
\Delta_3>1-\alpha,\qquad \textnormal{for }\alpha=\frac{(m_d+\epsilon)^2}{d(1-\delta)}
\]
 with probability at least 
\[
1-e^{-d^2\delta^2 /4}- e^{-d^2\epsilon^2}.
\]
By Proposition \ref{prop:fs}, we obtain 
\[
\Delta_\Lambda\geq 2(k-1) \left(\frac{1}{2k-2}-\alpha\right).
\]
Hence, we obtain a positive lower bound on $\Delta_\Lambda$ for all $k\geq 2$ and parameters $d,\delta,\epsilon$ satisfying the condition $2k-2<\alpha^{-1}$. Finally, we set $\epsilon=\delta/2$ and obtain Theorem \ref{thm:main1}.
\end{proof}

\subsection{Proof of the bound on $m_d$ (Proposition \ref{prop:md})}
\label{ssec:propmdproof}

We aim to prove that
    \begin{equation}\label{eq:largedestimate}
m_d^2=\min_{p\geq 1} \Ex\big[{\tr(G^{2p}})\big]^{1/(2p)}\leq 
    4e^2\log(d).
\end{equation}
We shall derive this from 
\begin{equation}\label{eq:Wick}
    \frac{1}{d} \Ex\big[{\tr(G^{2p}})\big]\leq \frac{(2p)!}{2^p p!},\qquad \forall p\in \mathbb Z_+.
\end{equation}
Indeed, assuming that \eqref{eq:Wick} holds, we choose $p_*=\lceil\log d\rceil$ and use Stirling's inequality  $(n/e)^n\leq n! \leq n^n$ to obtain
\begin{equation}
    m_d^2\leq \Ex\big[\tr(G^{2p_*})]^{1/p_*}\leq \left(\frac{(2p_*)!}{2^{p_*} p_*!}\right)^{1/p_*} \leq 2ep_* d^{1/p_*} \leq 2e^2 \lceil\log d\rceil. 
\end{equation}

It remains to prove \eqref{eq:Wick}, for which we use Wick's theorem. Let $p\in \mathbb Z_+$. We compute the $p$th moment by summing over Wick contractions and using $\Ex[G_{ij}G_{k\ell}] = \frac{1}{d}\delta_{ik}\delta_{j\ell}$ for the matrix elements. Let $M_{2p}$ be the subset of $S_{2p}$ consisting of pair partitions. We have
\begin{equation}
    \frac{1}{d} \Ex\big[{\tr(G^{2p}})\big] = \frac{1}{d} \Ex\big[{\tr(G^{\otimes 2p}}\Pi_{\rm cyc})\big] = \frac{1}{d}\sum_{\sigma\in M_{2p}} \frac{1}{d^p}\tr (\Pi_{\sigma}\Pi_{\rm cyc}) =  \frac{1}{d^{p+1}}\sum_{\sigma\in M_{2p}} d^{\ell(\sigma^{-1}{\rm cyc})}
\end{equation}
where $\Pi_{\rm cyc}$ is the cyclic permutation and $\ell(\sigma)$ is the length of the cycle type of the permutation $\sigma\in S_{2p}$. The number of cycles in a permutation product is related to the distance between two permutations $\sigma,\tau\in S_{2p}$ as $\ell(\sigma^{-1}\tau) = 2p-d(\sigma,\tau)$, where $d(\sigma,\tau)$ is the minimal number of transpositions required to take $\sigma$ to $\tau$. It follows from the triangle inequality that $\ell(\sigma^{-1}\tau) + \ell(\sigma) \leq 2p + \ell(\tau)$, and thus $\ell(\sigma^{-1}{\rm cyc}) \leq 2p + \ell({\rm cyc}) -\ell(\sigma) \leq p+1$ for $\sigma\in M_{2p}$, and so
\begin{equation}
    \frac{1}{d} \Ex\big[{\tr(G^{2p}})\big] = \frac{1}{d^{p+1}}\sum_{\sigma\in M_{2p}} d^{\ell(\sigma^{-1}{\rm cyc})}\leq \frac{(2p)!}{2^p p!},
\end{equation}
where the last inequality holds because  $\frac{(2p)!}{2^p p!}$ is the number of pair partitions of a set of length $2p$. This proves \eqref{eq:Wick} and hence Proposition \ref{prop:md}.
\qed

\begin{remark}
    There is an asymptotic expression for the moments of traces of $G$ \cite{brezin2000characteristic}
\begin{equation}
    \frac{1}{d} \Ex\big[{\tr(G^{2p}})\big] = (2p!) \sum_{\ell=0}^p \bigg(\prod_{j=0}^\ell\left(1-\frac{j}{d}\right)\bigg) \frac{1}{(2d)^{p-\ell}}\frac{1}{\ell!(\ell+1)!(p-\ell)!} = C_p + O(d^{-2})
\end{equation}
where the leading order term goes as the $p$-th Catalan number. With an exact control on the error term, this expression can be used to improve the  upper bound above for large $d$, but we prefer to give the  self-contained proof above counting the number of Wick contractions.
\end{remark}

\section{Deterministic Hamiltonians built from Haar projectors }\label{sec:2proof}
In this section, we consider the translation-invariant Hamiltonians \eqref{eq:HLambda2} whose nearest-neighbor interaction is defined in terms of Haar projectors,
\[
    P=Q^\perp,\qquad Q = \int_{\mathrm{U}(q^2)}  U^{\otimes t}\otimes \overline{U}^{\otimes t} \mathrm{d}\mu_H(U).
\]
We first recall the short proof of Lemma \ref{lm:2ff} on frustration-freeness from \cite{brown2010convergence}. Afterwards, we analytically verify the star-graph finite-size criterion Proposition \ref{prop:fs2} to obtain Theorem \ref{thm:main2}. 

\subsection{Proof of frustration-freeness}
\begin{proof}[Proof of Lemma \ref{lm:2ff}]

Note that $H_\Lambda\geq 0$. We explicitly write down zero-energy ground states, thereby proving frustration-freeness. Let $\{\ket{i}\}_{1\leq i\leq q}$ be an orthonormal basis of $\C^q$ and let $\{\ket{\vec{\imath}\,}\}$ be an orthonormal basis of $({\C^q})^{\otimes t}$, where $\vec\imath =(i_1,\ldots,i_t)\in \{1,\ldots,q\}^{t}$. For a permutation $\sigma\in S_t$, define the following vectors
\begin{equation}
    \ket{\psi_\sigma} = (\Pi_\sigma \otimes \iden) \ket\Omega\,, \where \ket\Omega = \frac{1}{q^{t/2}} \sum_{i_1,\ldots,i_t=1}^q \ket{\vec\imath\,}\otimes\ket{\vec\imath\,}\,,
\end{equation}
so $\ket\Omega\in ({\C^q})^{\otimes t}\otimes ({\C^q})^{\otimes t}$ is the normalized maximally entangled state between the $2t$ copies. The permutation operator $\Pi_\sigma$ is defined by its action on basis states as
\begin{equation}
    \Pi_\sigma \ket{i_1}\otimes\ldots\otimes \ket{i_t} = \ket{i_{\sigma(1)}}\otimes\ldots\otimes\ket{i_{\sigma(t)}}\,.
\end{equation}
Equivalently, $\ket{\psi_\sigma} $ is the vectorization of the standard representation of the permutation $\sigma\in S_t$. 

Since $\iden\otimes A^* \ket{\Omega} = A^\dagger \otimes \iden \ket{\Omega} $ for any operator $A$, we have $U^{\otimes t}\otimes U^*{}^{\otimes t} \ket{\psi_\sigma} = \ket{\psi_\sigma}$ for any unitary $U$ and any $\sigma\in S_t$. Therefore for all local terms $P_{x,y}$ we have $P_{x,y} \ket{\psi_\sigma}^{\otimes n} = 0$ for any $\sigma\in S_t$. 

This proves that for any graph $\Lambda$, the Hamiltonian is frustration free with zero-energy  ground states $\{\ket{\psi_\sigma}^{\otimes n}\}$. 
\end{proof}

We remark that the  states $\{\ket{\psi_\sigma}^{\otimes n}\}$ are not orthogonal, but they are almost orthogonal in the sense that overlap of two differing states is suppressed by dimension factors \cite{brandao2016local,Harrow23}. Nevertheless, the ground space of $H_\Lambda$ is spanned by the permutation states $\ker H_\Lambda = {\rm span}\{ \ket{\psi_\sigma}^{\otimes n}\}$ for all values of $n$ and $t$. 

\subsection{Lower bound on star graphs gaps}

In this subsection, we derive the requisite lower bound on the star graph gaps.
 \begin{proposition}\label{prop:star}
Let $n,q,t\geq 1$. For the model defined by \eqref{eq:HLambda2} and \eqref{eq:Qdefn} with $q>t$, the spectral gap on the $(n+1)$-site star graph is lower-bounded by
\begin{equation}
    \Delta_{n+1}\geq 1 - \frac{2nt^2}{q}.
\end{equation}
\end{proposition}

The proof idea is familiar from random quantum circuits on $n$ qudits of local dimension $q$ on a general graph $G$, such that the vertices of $G$ correspond to the qudits and the edges are the allowed 2-local interactions of the  architecture. Let the distribution $\nu_n^G$ be a single step of the evolution, i.e., a single layer of the circuit, defined as follows: pick an edge $(i,j) \in E(G)$ uniformly at random and apply a gate $U_{i,j}$ drawn randomly from the Haar measure on $U(q^2)$. Using the familiar large $q$ tricks, see e.g. \cite{haferkamp2021improved}, we can prove bounds on the star graph Hamiltonian gaps.

In the following we will repeatedly use the approximate orthogonality of permutation operators at large dimension. Specifically, as shown in \cite[Lemma 17(b)]{brandao2016local},
\begin{equation}\label{eq:17b}
    \big\| \sum_{\sigma\in S_t} \psi_\sigma^{\otimes n} - Q_{[n]}\big\|_\infty \leq \frac{t^2}{q^n}\,,
\end{equation}
where $Q_{[n]}$ is the ground state projector of $H_{\{1,\ldots,n\}}$. We will also use the operator $B$ which interchanges the permutation operators and an orthonormal basis for the kernel of the Hamiltonian
\begin{equation}
    B = \sum_{\sigma\in S_t} \ket{\psi_\sigma}\!\bra{\sigma}\,.
\end{equation}

We also note the following useful lemma.
\begin{lemma}\label{lm:psi} Let $\psi$ be a rank-one projector acting on $\C^D$ with $D\geq 2$. For any integer $m\geq 2$, we have
\begin{equation}
    \Big\|\sum_{i=1}^m (\psi)_i\otimes \iden_{[m]\setminus i} - m \psi^{\otimes m} \Big\|_\infty = m-1\,.
\end{equation}
\end{lemma}

\begin{proof}
First note that in the case $m=2$, we have
\begin{equation}
    \big\|\psi\otimes \iden + \iden\otimes\psi - 2\psi\otimes \psi\big\|_\infty = \big\|\psi\otimes \psi^\perp - \psi^\perp\otimes \psi\big\|_\infty = 1\,.
\end{equation}
In general, writing $\iden = \psi + \psi^\perp$ and expanding out, we find
\begin{align}
    \sum_{i=1}^m (\psi)_i\otimes \iden_{[m]\setminus i} - m \psi^{\otimes m} &= \sum_{\ell=1}^{m-1} \left( \ell\, \sum_{s\in \binom{[m]}{\ell}} \big(\psi^{\otimes\ell}\big)_s \otimes \big((\psi^\perp)^{\otimes (m-\ell)}\big)_{[m]\setminus s}\right),
\end{align}
where we use the notation that $\binom{[m]}{\ell}$ represents the set of all size-$\ell$ subsets of $[m]=\{1,\ldots,m\}$ and  $(\psi^{\otimes\ell})_s$ with $s\in \binom{[m]}{\ell}$ acts on the tensor factors in $s$, i.e.\ $(\psi^{\otimes 3})_{\{1,2,4\}} = (\psi)_1\otimes(\psi)_2\otimes(\psi)_4$.

We now observe that every term in the sum is a projector and that they are all mutually orthogonal. Therefore the operator $\sum_{i=1}^m (\psi)_i\otimes \iden_{[m]\setminus i} - m \psi^{\otimes m}$ is diagonal as written above, with eigenvalues $0,1,\ldots,m-1$. Thus, the operator norm is equal to $m-1$.
\end{proof}

\begin{proof}[Proof of Proposition \ref{prop:star}]
Using that the Hamiltonian gap can be rewritten as the norm difference of suitable moment operators as in \cite[Lemma 16]{brandao2016local}, we obtain
\begin{equation}
    \Delta_{n+1} = (n-1)\left(1- \Big\| \frac{1}{n-1}\sum_{i=1}^{n-1} Q_{i,n}\otimes \iden_{[n-1]\setminus \{i\}} - Q_{[n]}\Big\|_\infty\right)
\end{equation}
We bound the norm on the right-hand side  by using the approximate orthogonality of permutation operators at large $q$ shown in \eqref{eq:17b} and Lemma \ref{lm:psi},
\begin{align}
     &\Big\| \frac{1}{n-1}\sum_{i=1}^{n-1} Q_{i,n}\otimes \iden_{[n-1]\setminus \{i\}} - Q_{[n]}\Big\|_\infty\\
    &\leq \frac{1}{n-1}\Big\| \sum_{i=1}^{n-1}\sum_\sigma (\psi_\sigma)_i\otimes \iden_{[n-1]\setminus i}\otimes(\psi_\sigma)_n - (n-1)\sum_\sigma \psi_\sigma^{\otimes n} \Big\|_\infty + \frac{t^2}{q^2} + \frac{t^2}{q^n}\\
    &\leq \frac{1}{n-1}\big\|B^\dagger B\big\|_\infty \max_\sigma \Big\| \sum_{i=1}^{n-1} (\psi_\sigma)_i\otimes \iden_{[n-1]\setminus i} - (n-1) \psi_\sigma^{\otimes n-1} \Big\|_\infty + \frac{t^2}{q^2} + \frac{t^2}{q^n}\\
    &\leq \frac{n-2}{n-1}\left(1+\frac{t^2}{q}\right) + \frac{t^2}{q^2} + \frac{t^2}{q^n}\\
    &\leq \frac{n-2}{n-1} + \frac{2t^2}{q}.
\end{align}
Therefore, the star graph gap satisfies $\Delta_{n+1}\geq 1-\frac{2nt^2}{q}$, which proves Proposition \ref{prop:star}.
\end{proof}

\subsection{Conclusion}
\begin{proof}[Proof of Theorem \ref{thm:main2}]
Let $n\geq 2$. From Proposition \ref{prop:star}, we have
\[
\Delta_{n}\geq 1-\frac{2(n-1)t^2}{q}.
\]
Recall also that $\Delta_{2}=1$, trivially. 

Let $\Lambda$ be a graph of maximal degree $ k$.
By Proposition \ref{prop:fs2}, we conclude
\[
 \Delta_\Lambda \geq 2\bigg( \min_{2\leq m\leq k} \Delta_{m} - \frac{1}{2}\bigg)
\geq 1-\frac{4(k-1)t^2}{q}
\]
as desired.
\end{proof}

\subsubsection*{Acknowledgments}
We would like to thank Dan Eniceicu, C\'{e}cilia Lancien, Shivan Mittal, David Perez-Garcia for helpful discussions. 
NHJ would like to acknowledge the Centro de Ciencias de Benasque Pedro Pascual for hospitality during the completion of part of this work. NHJ acknowledges support in part from DOE grant DE-SC0025615. The research of ML is  supported by the DFG
through the grant TRR 352 – Project-ID 470903074 and by the European Union (ERC Starting Grant MathQuantProp, Grant Agreement 101163620).\footnote{Views and opinions expressed are however those of the authors only and do not necessarily reflect those of the European Union or the European Research Council Executive Agency. Neither the European Union nor the granting authority can be held responsible for them.}

\appendix

\begin{footnotesize}
\bibliographystyle{plain}
\bibliography{refs}
\end{footnotesize}

\end{document}